\documentclass[11pt, a4paper]{article}
\usepackage{psvectorian}
\usepackage{graphicx}
\usepackage{amssymb}
\usepackage{enumitem}
\usepackage{geometry}
\usepackage{subfig}
\usepackage{tikz}
\usetikzlibrary{arrows,shapes,positioning}
\usepackage{float}
\usepackage{url}
\usepackage{mathtools}
\usepackage{amsmath,amsfonts,amsthm}
\usepackage{mathrsfs, amsmath} 

\newcommand{\case}[1]{\paragraph*{Case #1:}}

\newtheorem{theorem}{Theorem}[section]

\newtheorem{conjecture}[theorem]{Conjecture}

\newtheorem{example}{Example}

\usepackage{subfig}

\usepackage{subfig}
\usepackage[colorlinks,
linkcolor=blue,
anchorcolor=blue,
citecolor=blue
]{hyperref}
\begin{document}
\begin{center}
{\large \bf The Role of Mathematical Folk Puzzles in Developing mathematical Thinking and Problem-Solving Skills}
\end{center}
\begin{center}
  Duaa Abdullah$^{1}$  \hspace{0.2 cm} Jasem Hamoud$^
  {2}$\\[6pt]
  $^{1,2}$ Physics and Technology School of Applied Mathematics and Informatics \\
Moscow Institute of Physics and Technology, 141701, Moscow region, Russia\\[6pt]
Email: $^{1}${\tt duaa1992abdullah@gmail.com}, $^{2}${\tt hamoud.math@gmail.com}
\end{center}
\noindent
\begin{abstract}
This paper covers a variety of mathematical folk puzzles, including geometric (Tangrams, dissection puzzles), logic, algebraic, probability (Monty Hall Problem, Birthday Paradox), and combinatorial challenges (Eight Queens Puzzle, Tower of Hanoi). It also explores modern modifications, such as digital and gamified approaches, to improve student involvement and comprehension. Furthermore, a novel concept, the "Minimal Dissection Path Problem for Polyominoes," is introduced and proven, demonstrating that the minimum number of straight-line cuts required to dissect a polyomino of N squares into its constituent units is $\mathrm{N}-1$. This problem, along with other puzzles, offers practical classroom applications that reinforce core mathematical concepts like area, spatial reasoning, and optimization, making learning both enjoyable and effective.
\end{abstract}

\noindent\textbf{AMS Classification 2010:} 397C70, 97U40, 97F90 , 97D60, 97D70, 97A40.


\noindent\textbf{UDC:} 37.013,51

\section{Introduction}

Mathematical folk puzzles in many cultures and for many centuries have been a teaching and learning tool for schoolchildren. In the 21st century, school children from different countries, enjoy doing these puzzles. The puzzles provide an intriguing and informal context for the children to learn or explore a wealth of mathematical concepts and skills, including geometrical shapes, area and perimeter, symmetry, angles, patterns and relationships, and multiplicative thinking. Yet, adults who attend to children often ignore these explorations and learning opportunities, seeing the child's engagement only as play. In a school culture, interest in mathematical puzzles as a challenge and competition may lead to achievement.

Math puzzles are a valuable and effective educational tool for teaching and learning math concepts among students of all ages. These captivating puzzles not only provide a fun and engaging experience but also encourage critical thinking and problem-solving skills. They offer an exciting opportunity for children to explore and develop various mathematical skills in a playful manner. The benefits of incorporating math puzzles into the classroom are numerous. First and foremost, puzzles promote critical thinking and logical reasoning skills. They need students to analyze and strategize, enabling them to think creatively and develop innovative problemsolving techniques. Through the process of solving puzzles, students learn to think outside the box, approaching mathematical challenges from different perspectives.

As well, math puzzles encourage perseverance and resilience. As students encounter challenging puzzles, they develop patience and persistence in their pursuit of solutions. They learn to embrace the process of trial and error, understanding that mistakes are an essential part of the learning experience. This resilience nurtures a growth mindset, where students view challenges as opportunities for growth and improvement. In addition to cognitive benefits, math puzzles also enhance collaboration and teamwork. It promotes a positive and inclusive learning environment, where students support and learn from one another.

By tapping into the innate curiosity and playfulness of children, math puzzles make learning an enjoyable and fulfilling experience. They stimulate interest and motivation, transforming the perception of mathematics from dry and monotonous to exciting and captivating. As students become immersed in the world of math puzzles, they develop a genuine love for the subject and an eagerness to explore its endless possibilities. So, we present some puzzles that may engender interest, challenge, and mathematical commitment, all in an enjoyable context.

\section{Historical Background of Mathematical Folk Puzzles}
The historical background of mathematical folk puzzles reveals their deep-rooted presence in various cultures, showcasing the universal appeal of mathematics. The educational value of these puzzles is significant, as they encourage exploration and creativity while reinforcing mathematical principles~\cite{Boaler}.
Mathematical folklore, like any good folklore, always reflects both historical and cultural contexts in which folk (or professional) mathematicians lived. Thus, these puzzles give us a glimpse of educational practices and mathematical attitudes of professional and non-professional mathematicians in previous generations. Professional mathematicians have always had problems that attracted their interest, originally those which could be at least approximately solved by geometrical methods. On the other hand, folk mathematicians have been concerned with another class of problems those for which the answer is obtained by counting. It seems that such puzzles, relating to real-life situations, were not simply for amusement purposes but were useful. They were everyday puzzles which were solved mentally summing sticks in their hands or calculating the capacity of vessels were of vital importance in peasant life~\cite{Picciotto}.

In the process of solving mathematical folk puzzles, properties of various objects and numbers are considered and, anyway, neglected in school curricula, and the history of mathematics is taught without examples. Moreover, it makes children think and even create themselves. Finally, on the one hand, schoolchildren can often use folk puzzles and the investigation of their new and then forgotten ideas for studying the school curriculum. It can unite topics usually related to one another in geography classes or physics classrooms. The majority of mathematical folk puzzles originate from geographical or household problems \cite{Danesi,Iortser}.

\section{Educational Value and Role of puzzles in learning}
While solving mathematical puzzles or problems posed by a teacher, in particular those of a "folk" nature, children not only learn a certain mathematical concept and its properties. They often apply their knowledge acquired during classes, encountering mathematical tasks formulated in the form of games or puzzles, and solving fairly complicated problems as well. For example, during a lesson, children learn to determine the denominators equal to those among the fractions with the numerators $1,1 / 2,1 / 3$, and also some of the problems such as ``Plenty of Ones'', ``What Was Scaled?'', ``One Summand of the Sum is Known'', or others containing a trivial task as a part of the puzzle. The advantage of so-called "folk problems" is that they often clearly indicate what measures or operations should be made for the problem to be solved.

An important characteristic of puzzles is their close relation in comparison with other problems of a logical or mathematical nature. While solving puzzles, one often obtains logical skills that are sufficient for solving problems of much interest, in particular. Such ability can and should be developed in both boys and girls. After all, it can be done in the form of a game, and in an interesting way. Along such lines, the tradition of solving school puzzles has developed, which was borrowed from folklore, and puzzles prepared with pedagogical goals have long existed, which were not only a test of psychic abilities but also a means to develop them.

In fact, puzzles are increasingly being recognized as an important teaching and learning strategy. In this light, teachers should strive to incorporate puzzles in their day-to-day activities, not just in mathematics, but in other subjects as well. This, combined with the use of digital interactive whiteboards and websites, can contribute to better student engagement, understanding, and attainment. The specific focus of this research is to explore and evaluate what the key purposes, principles, and strategies are when using puzzles as a teaching and learning tool in a lesson.

In today's educational environment, students are faced with the task of solving various types of problems. However, mathematics consistently ranks as one of the most challenging subjects for teachers in terms of helping students develop their problem-solving skills, regardless of their academic level or age group \cite{Pt}. The utilization of puzzles in mathematics instruction is not a recent phenomenon, as it has been employed for centuries. Numerous studies have demonstrated that puzzles are an excellent means of not only enhancing knowledge, but also fostering a genuine enjoyment of math lessons among students \cite{Aldila}.

Puzzles play a crucial role in developing higher-order thinking skills. These enigmatic problems cannot be solved through quantitative research or mere routine analysis, but rather require determination and logical reasoning. According to Sukuraman's research, a puzzle can be defined as a perplexing problem that necessitates the use of one's reasoning abilities to find a solution, rather than simply considering the subject matter at hand. Sukuraman further asserts that one of the primary reasons why puzzles are highly motivating for students is that they naturally captivate their attention. Puzzles arouse curiosity and instill a desire to persevere until the correct solution is reached \cite{Ahdhianto}.

\section{Geometric Puzzles Types}
Geometric puzzles come in a wide range of shapes, demanding your spatial reasoning, visual skills, and comprehension of geometric principles. Here's a breakdown of the most frequent types:

\subsection{Tangrams}
It is an ancient Chinese puzzle, consisting of a seven-piece ( 5 triangles, 1 square and 1 parallelogram); flat puzzle with tans that can be fitted together to create different shapes (see Figure~\ref{fign1}). Students can experiment with differences and discover geometric features with great curiosity and fun.

\begin{figure}[H]
\centering
  \includegraphics[width=\textwidth]{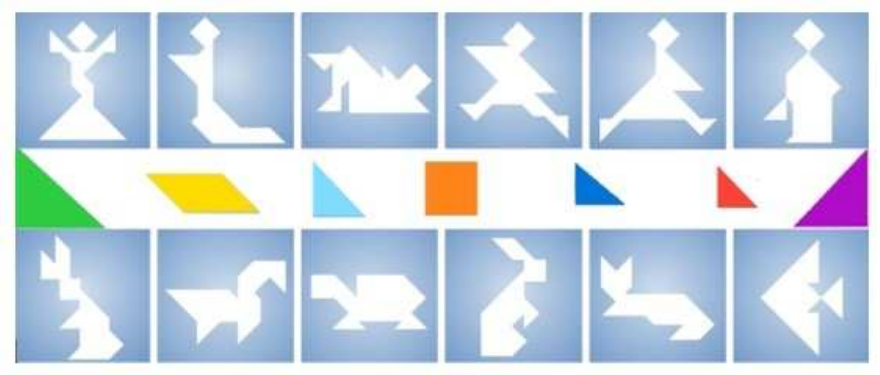}
\caption{ The dissection of a square into a triangle.}~\label{fign1}
\end{figure}

Tangram puzzles are commonly used as an educational tool for both youngsters and adults. They provide a hands-on approach to understanding geometric concepts like rotation, symmetry, and congruence. Players improve their awareness of spatial relationships and problem-solving skills by manipulating the seven tans into various shapes. Experiential learning can help consolidate mathematical concepts, making them more accessible and understandable.

Tangrams not only improve mathematics skills, but they also aid to develop other vital competencies. These involve fine motor abilities, as moving the small components necessitates accuracy and control. Furthermore, tangram puzzles can help improve concentration and patience because they often demand focus and repeated efforts. Tangrams are often used in classrooms and educational curricula due to improving these skills of geometric thinking, problem solving and spatial reasoning.

\subsection{Dissection Puzzles and A Visual Delight}
Dissection puzzles are geometric puzzles where a shape is cut into pieces that can be rearranged to form a different shape. These puzzles are not only fun but also a great way to develop spatial reasoning and problem-solving skills.

An example of this is the dissection of a square into a triangle (see Figure~\ref{fign2}); this is a puzzle that involves dissecting a square into several pieces that can be rearranged to form a triangle.
\begin{figure}[H]
\begin{center}
  \includegraphics[width=10cm, height=5cm]{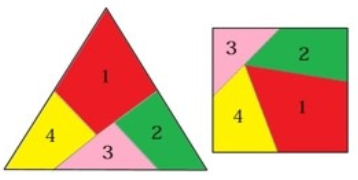}
\caption{Tangram puzzles of people or animals.}~\label{fign2}
\end{center}
\end{figure}

Another is the pentomino puzzle, which is a puzzle that includes twelve shapes, each consisting of five squares connected from edge to edge. The goal is to arrange the pentominoes to form different shapes, such as rectangles or squares.

\subsection{ Geometric Jigsaw Puzzles}
These puzzles consist of pieces that fit together to form a specific shape or image. They help students, through hands-on learning, understand how to combine and manipulate different shapes to form a suitable shape or pattern. We also add These puzzles develop problem-solving skills through critical thinking as children must figure out how each piece fits together to create the complete picture, which requires patience, time, perseverance, and creativity and imagination to complete the challenge, which enhances spatial awareness and encourages children to think outside the box.

Another type is called Mosaic Jigsaw Puzzles which uses tiles or coloured paper, students can create mosaic designs. This activity helps them understand tiling and the concept of area. Overall, jigsaw puzzles are an effective and fun tool in education.

\subsection{Pattern Blocks}
Pattern blocks are colourful geometric shapes in six different colours, usually made of wood or plastic. Small and lightweight, they enable students to create different shapes (triangles, squares, and hexagons) to create patterns or designs that reinforce the concepts of symmetry and space.

Pattern blocks are also a fun educational tool that students are naturally drawn to, and can be used to teach concepts in math, geometry, and art. They build fine motor skills, enhance creativity, and stimulate hands-on exploration.

\section{ Digital and Gamified Approaches}
Integrating technology into education has a significant positive impact in addition to enhancing engagement and learning outcomes. There are many methods, including:

\subsection{Gamified Learning}
Studies show that gamified teaching methods significantly boost motivation among primary school students. Incorporating game-like elements into geometry lessons can make learning more enjoyable and effective \cite{Luo}. Gamified learning environments make learning more fun and interactive due to their challenges, reward system, and possible storytelling, thus stimulating positive behavior and increasing motivation \cite{Cal}.

\subsection{Augmented Reality (AR)}
AR can be leveraged to create realistic interactive experiences with real-world environments that are richer, making abstract concepts tangible through virtual dissection or reenactment of historical events, ensuring that students remember information for longer, and this approach has been shown to significantly improve the understanding of complex engineering concepts \cite{Flores}.

\section{Classroom Activities}
Here are some classroom activities that can be implemented in the classroom, that focus on shapes, spatial thinking, design, and geometry:
\begin{enumerate}
    \item \textbf{Shape Hunt:} Have students search for objects that match certain geometric shapes and then name them out loud. Provide students with a shape search recording sheet where they can write and draw the shapes they find.
    \item \textbf{Design a Building:} Using digital tools or physical materials, students can act as architects to design a building using specific geometric shapes, enhancing their spatial reasoning skills \cite{Puig}. As you can involve students in designing their ideal learning space, they can brainstorm ideas for a classroom that is right for them.
    \item \textbf{Geometry Art:} Encourage students to create their own artwork; using geometric patterns and shapes. This integrates creativity with mathematical concepts.
\end{enumerate}

\section{Logic Puzzles and Algebraic Puzzles}
Logic puzzles and algebraic puzzles are fun challenges that require mathematical knowledge, deductive reasoning, and problem-solving skills. Logic puzzles often involve finding creative solutions to difficult situations by considering multiple perspectives and exploring unexpected connections between elements. They require little general knowledge and rely on basic mathematical principles. Two common types of logic puzzles are River Crossing Puzzles and Knights and Knaves Puzzles. There is a uniqueness in the scenarios of each type and the rules that it has, which makes participants solve them. These puzzles provide great settings for practicing critical thinking and also challenge the student, solving them to think logistically along with a strategy that adheres to specific rules \cite{Rakshit}.

In addition, they are a fun and powerful educational tool that builds problem-solving skills for all ages. Algebraic puzzles may involve variables like algebraic equations, and solving them requires advanced applications and mathematics such as equations and calculus. Both equation puzzles and function puzzles serve as valuable educational tools that enhance mathematical reasoning, and these puzzles not only test mathematical skills but also encourage creative thinking and logical deduction.

They challenge individuals to apply their knowledge creatively while developing critical thinking skills essential for advanced mathematics. Engaging with these puzzles can lead to improved problem-solving abilities and a deeper appreciation for algebra \cite{van}.

Equation puzzles require solvers to find unknown values that satisfy given mathematical equations. Common examples include solving for x in equations like $2 \mathrm{x}+3=11$, and another like Magic Squares need to arrange numbers in a square grid so that the sums of each row, column, and diagonal are equal. These techniques encourage a deeper understanding of algebraic principles and enhance problem-solving skills \cite{Malavolta}. Function puzzles involve understanding the behaviour of mathematical functions, often requiring solvers to determine outputs based on given inputs or vice versa. Examples of function composition need finding $f(g(x))$ where $f(x)$ and $g(x)$ are defined functions \cite{Vadim}. Function puzzles promote a robust understanding of mathematical concepts, encouraging learners to explore relationships between variables in innovative ways \cite{Imperio}.

\section{Probability Puzzles}
Probability puzzles often reveal counterintuitive results and challenge our understanding of probability theory. Two well-known examples are:

\begin{enumerate}
  \item Monty Hall Problem is a classic probability puzzle based on a game show scenario where you are presented with three doors; behind one door is a car (the prize), and behind the other two doors are goats. You have to choose one door, then host Monty Hall, who knows where the car is, opens one of the unchosen doors to reveal a goat. Monty then offers the contestant the option of switching to the other unopened door.
  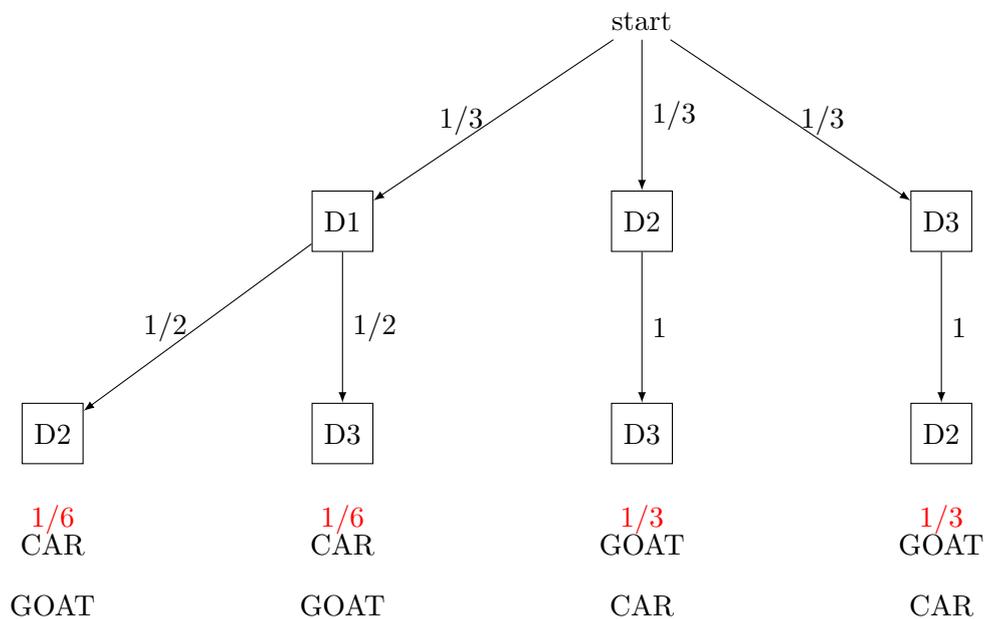
\begin{figure}[H]
\centering
\begin{tikzpicture}[scale=.6, node distance=20mm and 30mm,
  box/.style={rectangle, draw, minimum size=8mm, text centered},
  >=latex
]
\node (start) {start};

\node[box, below left=of start] (D1) {D1};
\node[box, below=of start] (D2) {D2};
\node[box, below right=of start] (D3) {D3};

\node[box, below left=of D1] (D2b) {D2};
\node[box, below=of D1] (D3b) {D3};

\node[box, below=of D2] (D3c) {D3};

\node[box, below=of D3] (D2c) {D2};

\draw[->] (start) -- node[left] {1/3} (D1);
\draw[->] (start) -- node[right] {1/3} (D2);
\draw[->] (start) -- node[right] {1/3} (D3);

\draw[->] (D1) -- node[left] {1/2} (D2b);
\draw[->] (D1) -- node[right] {1/2} (D3b);

\draw[->] (D2) -- node[right] {1} (D3c);

\draw[->] (D3) -- node[right] {1} (D2c);

\node[below=4mm, red] at (D2b.south) {1/6};
\node[below=4mm, red] at (D3b.south) {1/6};
\node[below=4mm, red] at (D3c.south) {1/3};
\node[below=4mm, red] at (D2c.south) {1/3};

\node[below=8mm] at (D2b.south) {CAR};
\node[below=8mm] at (D3b.south) {CAR};
\node[below=8mm] at (D3c.south) {GOAT};
\node[below=8mm] at (D2c.south) {GOAT};

\node[below=16mm] at (D2b.south) {GOAT};
\node[below=16mm] at (D3b.south) {GOAT};
\node[below=16mm] at (D3c.south) {CAR};
\node[below=16mm] at (D2c.south) {CAR};

\end{tikzpicture}
\caption{Decision tree Monty Hall problem}~\label{fign3}
\end{figure}
The solution to this problem (see Figure~\ref{fign3}) is as follows: Counterintuitively, the best strategy is to always switch doors. Because when you first choose a door, there's a $\frac{1}{3}$ chance that you picked the car and a $\frac{2}{3}$ chance that you picked a goat. If you switch after Monty reveals a goat, you win if your initial choice was wrong (which happens $\frac{2}{3}$ of the time).\par 
Thus, by switching, your probability of winning the car increases to $\frac{2}{3}$, while sticking with your initial choice gives you only $\frac{1}{3}$ chance of winning.
\item Birthday Paradox refers to the finding that in a group of n randomly picked persons, the likelihood that at least two people share the same birthday exceeds $50 \%$ with surprisingly few individuals. Surprisingly, just 23 people are required for this chance to exceed $50 \%$.
\end{enumerate}

The ``paradox'' is that one would expect a far larger number of persons to be required before the likelihood of sharing a birthdate becomes substantial \cite{Alvarado}. The reasoning is the paradox arises from calculating the probability of no shared birthdays:
\begin{enumerate}
  \item For the first person, there are 365 choices (assuming no leap year).
  \item For the second person, there are 364 choices to avoid matching the first person's birthday.
  \item Continuing this way for n people gives:
\end{enumerate}
\[
\overline{\mathrm{p}}(\mathrm{n})=1 \times\left(1-\frac{1}{365}\right) \times\left(1-\frac{2}{365}\right) \times \ldots \times\left(1-\frac{\mathrm{n}-1}{365}\right)=\prod_{\mathrm{n}=2}^{\mathrm{n}-1}\left(1-\frac{\mathrm{n}-1}{365}\right) .
\]
The probability that at least two people share a birthday is $\mathrm{p}(\mathrm{n})=1-\overline{\mathrm{p}}(\mathrm{n}) .$ Approximating the probability of no shared birthday:
$$
\overline{\mathrm{p}}(\mathrm{n}) \approx\left(\frac{364}{365}\right)^{\binom{\mathrm{n}}{2}} .
$$
\noindent  
Therefore, the probability of at least one shared birthday:

$$
\mathrm{p}(\mathrm{n}) \approx 1-\left(\frac{364}{365}\right)^{\binom{\mathrm{n}}{2}} .
$$
\noindent 
For 23 people, the probability of sharing a birthday is approximately 0.500477 \cite{Feliciano}. This startling conclusion demonstrates how human perception frequently underestimates probabilities in combinatorial circumstances \cite{Bevan}. The birthday problem has practical uses, such as detecting cryptographic attacks and determining the danger of hash collisions. Both the Monty Hall Problem and the Birthday Paradox illustrate how human intuition can often misjudge probabilities. The Monty Hall Problem shows that switching choices can lead to better outcomes, while the Birthday Paradox reveals how shared events can be more likely than expected in groups. Two examples for teaching concepts in probability and statistics, revealing how our intuitions can lead us astray in understanding likelihoods and outcomes.

\section{Combinatorial Puzzles}
Combinatorial puzzles are a type of mathematical challenge that involves arranging, selecting, and manipulating discrete items. They are enjoyable for both recreational and logical learners. These are typically combinatorial problems that demand thought to solve. Classic subcategories in this discipline include chessboard problems and permutation puzzles.

A combinatorial puzzle entails mixing distinct elements or pieces in a variety of ways to get a desired outcome or solution. These puzzles frequently use dissimilar and non-symmetrical parts to increase the number of possible combinations. The goal for puzzle designers is to produce simple puzzles with few components that are both engaging and puzzling, ideally with only one correct combination.

\subsection{Chessboard Problems}
Chessboard problems classification Chessboard problems are a class of combinatorial problems that make use of the structure and properties of a chessboard. Knight's tours, N queen's problems, etc. While these problems usually consist of putting the chess pieces on the board in a particular configuration or determining how many arrangements are possible under specific limitations. Some common examples include:

\begin{enumerate}
  \item The Eight Queens Puzzle: Arrange eight queens on an $8 \times 8$ chessboard so that none threaten one another.
  \item The Knight's Tour: Given a chessboard and a knight, find a sequence of moves for the knight that will visit every square exactly once.
  \item Domination Problems: The goal is to find the minimum number of your pieces required to attack or control all squares on the board. They tend to be very tricky, featuring intricate chess piece movements and solution strategies.
\end{enumerate}

\subsection{Permutation Puzzles}
Permutation puzzles are also a significant subtype of combinatorial puzzles that involve rearranging a finite set of objects into a particular order or configuration. Many logic puzzles have a huge number of possible states and require a systematic approach to be solved. A few popular ones include:

\begin{enumerate}
  \item The Rubik's Cube: A 3D combination puzzle, invented in 1974 and gaining worldwide popularity in the 1980s, its objective to twist and turn the cube of smaller cubes back to original configuration with each face one solid colour.
  \item Sliding Puzzle: A flat board with numbered tiles that can be slid into an empty space to arrange them in an order.
  \item Combination Locks: Devices that must be opened with a series of numbers or symbols. Many permutation puzzles have ties to group theory and can be studied with mathematical techniques to determine solution strategies and lower bounds on the length of a solution.
  \item Tower of Hanoi: This a classic problem that involves moving a stack of disks from one rod to another, while obeying certain rules. A puzzle consists three rods and a set of disks of different size. The objective is to transfer the entire stack to another rod following these specific rules.
\end{enumerate}

The number of movements required at least is $2 \mathrm{n}-12 \mathrm{n}-1$, where $n$ is the number of disks. The puzzle demonstrates recursion and problem-solving strategies. All of these puzzles not only entertain but also help develop problem-solving skills and logical thinking abilities.

\section{The Minimal Dissection Path Problem for Polyominoes}
While traditional dissection puzzles focus on rearranging pieces to form new shapes, a novel area of inquiry lies in the efficiency of the dissection process itself. We introduce the ``Minimal Dissection Path Problem for Polyominoes'', which seeks to determine the minimum number of straight-line cuts required to transform a given polyomino into a set of its constituent unit squares. This problem has implications for understanding the inherent structural complexity of polyominoes and could inform the design of new types of puzzles where the challenge lies not in rearrangement, but in efficient deconstruction.

Consider a polyomino P composed of N unit squares. A straight-line cut is defined as a cut that extends from one edge of the polyomino to another, or from an edge to an internal point, effectively dividing the polyomino into two or more smaller pieces. The goal is to reduce P to N individual unit squares using the fewest possible cuts.

\begin{conjecture}~\label{conj1}
For any polyomino P with N unit squares, the minimum number of straight-line cuts required to dissect P into N individual unit squares is at least $\mathrm{N}-1$.
\end{conjecture}
\begin{proof}
Let $P$ be a polyomino composed of $N$ unit squares. Our goal is to dissect $P$ into $N$ individual unit squares using the minimum number of straight-line cuts. A straight-line cut is defined as a cut that divides a piece of the polyomino into two or more smaller pieces. Each cut effectively increases the number of disconnected pieces by at most one.

\case{1} \textbf{Lower Bound ( $k \geq N-1$ ):} 
Consider the initial state where we have one connected polyomino, $P_{0}$, which represents a single piece. Let $C_{i}$ be the number of connected pieces after $i$ cuts. Initially, $C_{0}=1$. When a straight-line cut is made, it can pass through one or more existing pieces. However, each single cut can, at most, increase the total number of distinct pieces by one. For example, if a cut passes through a single piece, it divides that piece into two, increasing the total count by one. If a cut passes through multiple pieces simultaneously (e.g., a long cut through a row of squares), it still effectively separates each piece it traverses into two, but the net increase in the number of disconnected components is still at most one per cut, as a single cut cannot create more than one new boundary that separates previously connected parts \cite{Demaine}.

Therefore, after $k$ cuts, the maximum number of pieces, $C_{k}$, can be expressed as:

$$
C_{k} \leq C_{0}+k C_{k} \leq 1+k
$$

To dissect the polyomino $P$ into $N$ individual unit squares, we must achieve a state where there are $N$ distinct pieces. Thus, we require $C_{k}= N$. Substituting this into the inequality:

$$
N \leq 1+k k \geq N-1
$$

This proves that at least $N-1$ straight-line cuts are necessary to separate a polyomino of $N$ unit squares into $N$ individual unit squares. This lower bound is intuitive: to obtain $N$ items from one, you need to make at least $N-1$ separations.
\case{2} \textbf{Sufficiency (Can $N-1$ cuts always be achieved?)}
The sufficiency of $N-1$ cuts depends on the geometry of the polyomino. While $N-1$ cuts are necessary, they are not always sufficient if the cuts cannot be made in a way that isolates individual squares without intersecting previously made cuts in a complex manner that prevents further simple separation. However, for any polyomino, it is always possible to achieve $N-1$ cuts by strategically making cuts along the grid lines that define the unit squares.

Consider any polyomino $P$ composed of $N$ unit squares. Each unit square is defined by a grid of horizontal and vertical lines. To separate all $N$ squares, we need to make cuts along these internal grid lines. A polyomino with $N$ squares has $N-1$ internal boundaries between adjacent squares if it is connected. For example, a $1 \times N$ polyomino (a straight line of $N$ squares) has $N-1$ internal vertical boundaries. Each of these boundaries can be severed with a single straight-line cut, resulting in $N$ individual squares using exactly $N-1$ cuts.

For more complex polyominoes, the strategy involves systematically isolating each square. Any internal edge between two adjacent unit squares within the polyomino can be severed by a straight-line cut. If we consider the dual graph of the polyomino, where each square is a node and an edge exists if two squares share a boundary, then dissecting the polyomino into individual squares is equivalent to removing all edges in this dual graph. A minimal set of cuts corresponds to a minimum edge cut that disconnects all nodes. Since the dual graph of a connected polyomino is a connected graph with $N$ nodes, it must have at least $N-1$ edges. Each cut can effectively remove one or more of these internal edges. By making cuts along the grid lines that separate adjacent squares, we can ensure that each cut contributes to isolating at least one square or separating a larger piece into two smaller ones. Since each cut can at most increase the number of pieces by one, and we start with one piece and want $N$ pieces, we need exactly $N-1$ cuts to achieve this if each cut is optimally placed to increase the piece count by one. This is always possible by following the internal grid lines. For example, one can always make a cut that separates an 'end' square from the rest of the polyomino, reducing the problem to dissecting a smaller polyomino. Repeating this process $N-1$ times will yield $N$ individual squares.

Therefore, for any polyomino $P$ with $N$ unit squares, $N-1$ straight-line cuts are both necessary and sufficient to dissect $P$ into $N$ individual unit squares.
So the minimum number of straight-line cuts required to dissect a polyomino of $N$ unit squares into $N$ individual unit squares is exactly $N-1$.
\end{proof}

\subsection*{Examples of Minimal Dissection Paths}

To illustrate Conjecture 1 and its proof, let us consider a few specific polyomino examples and their minimal dissection paths.

\begin{example}[\textbf{The Domino ( $\mathbf{N} \boldsymbol{=} \mathbf{2}$ )}]~\label{ex1}
A domino is a polyomino composed of two unit squares. According to Conjecture 1, the minimum number of cuts required to dissect a domino into two individual unit squares is $N-1=2-1=1$ cut. 
\end{example}

As shown in Figure~\ref{fign4}, a single straight-line cut along the shared edge of the two squares is sufficient to separate them. This visually confirms the conjecture for the simplest case.

\begin{figure}[H]
\centering
\begin{tikzpicture}[scale=.8]
\draw [line width=2pt] (0,6)-- (0,4);
\draw [line width=2pt] (0,4)-- (2,4);
\draw [line width=2pt] (2,4)-- (2,6);
\draw [line width=2pt] (2,6)-- (0,6);
\draw [line width=2pt] (2,6)-- (4,6);
\draw [line width=2pt] (4,6)-- (4,4);
\draw [line width=2pt] (4,4)-- (2,4);
\draw [line width=2pt] (0,3)-- (2,3);
\draw [line width=2pt] (2,3)-- (2,1);
\draw [line width=2pt] (2,1)-- (0,1);
\draw [line width=2pt] (0,1)-- (0,3);
\draw [line width=2pt,dash pattern=on 1pt off 1pt] (3,3)-- (3,1);
\draw [line width=2pt] (4,3)-- (4,1);
\draw [line width=2pt] (4,3)-- (6,3);
\draw [line width=2pt] (6,3)-- (6,1);
\draw [line width=2pt] (6,1)-- (4,1);
\end{tikzpicture}
\caption{Minimal Dissection of a Domino. A $1 \times 2$ domino polyomino is shown, with a single straight line cut separating it into two individual unit squares.}~\label{fign4}

\end{figure}
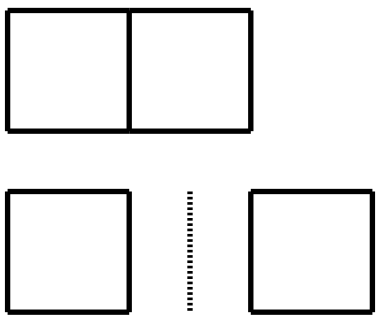

\begin{example}[The L-Tromino ( $\mathbf{N} \boldsymbol{=} \mathbf{3}$ )]~\label{ex2}
An L-tromino is a polyomino composed of three unit squares arranged in an 'L' shape. For an L-tromino, Conjecture 1 predicts that $N-1= 3-1=2$ cuts are necessary and sufficient to dissect it into three individual unit squares.
\end{example}

Figure~\ref{fign5} illustrates how two strategic straight-line cuts can achieve this. The first cut can separate one of the end squares, leaving a domino. The second cut then separates the remaining two squares. The specific placement of cuts can vary, but the total number of cuts remains two.
\begin{figure}[H]
\centering
\begin{tikzpicture}[scale=.8]
\draw [line width=2pt] (0,6)-- (0,4);
\draw [line width=2pt] (0,4)-- (2,4);
\draw [line width=2pt] (2,4)-- (2,6);
\draw [line width=2pt] (2,6)-- (0,6);
\draw [line width=2pt] (2,6)-- (4,6);
\draw [line width=2pt] (4,6)-- (4,4);
\draw [line width=2pt] (4,4)-- (2,4);
\draw [line width=2pt] (0,3)-- (2,3);
\draw [line width=2pt] (2,3)-- (2,1);
\draw [line width=2pt] (2,1)-- (0,1);
\draw [line width=2pt] (0,1)-- (0,3);
\draw [line width=2pt] (4,3)-- (4,1);
\draw [line width=2pt] (4,3)-- (6,3);
\draw [line width=2pt] (6,3)-- (6,1);
\draw [line width=2pt] (6,1)-- (4,1);
\draw [line width=2pt] (0,6)-- (0,8);
\draw [line width=2pt] (0,8)-- (2,8);
\draw [line width=2pt] (2,8)-- (2,6);
\draw [line width=2pt] (7,3)-- (9,3);
\draw [line width=2pt] (9,3)-- (9,1);
\draw [line width=2pt] (9,1)-- (7,1);
\draw [line width=2pt] (7,1)-- (7,3);
\end{tikzpicture}
\caption{Figure 5. Minimal Dissection of an L-Tromino. An L-tromino polyomino is shown, with two straight line cuts separating it into three individual unit squares.}~\label{fign5}
\end{figure}
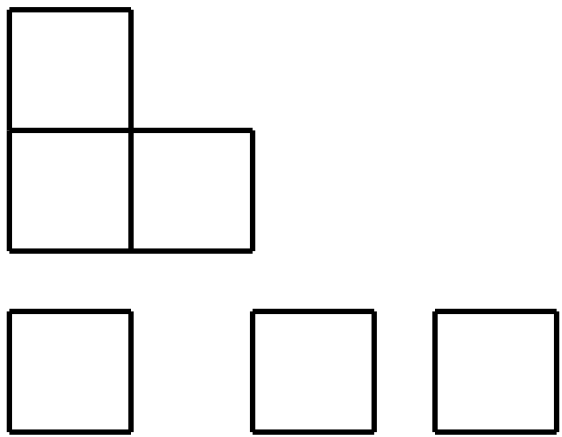

These examples~\ref{ex1} and \ref{ex2} demonstrate the practical application of the Minimal Dissection Path Problem and visually reinforce the validity of Conjecture 1. The simplicity of these cases allows for clear understanding of the cutting process and the resulting individual unit squares.

\subsection{Relevance and Adaptation for Classroom Use}
The Minimal Dissection Path Problem for Polyominoes, while rooted in theoretical mathematics, offers significant pedagogical value and can be effectively adapted for various classroom settings, from elementary to advanced levels. Its hands-on nature, combined with the underlying mathematical principles, makes it an engaging tool for fostering critical thinking, spatial reasoning, and problem-solving skills.

\subsubsection{Pedagogical Benefits}
\begin{enumerate}
  \item \textbf{Conceptual Understanding of Area and Decomposition:} At a fundamental level, dissecting polyominoes into unit squares reinforces the concept of area as a sum of individual units. Students can visually and kinesthetically understand that a polyomino with $N$ squares has an area of $N$ square units, and that these units can be separated \cite{Boaler}.
  \item \textbf{Spatial Reasoning and Visualization: }The process of identifying optimal cut lines requires strong spatial reasoning. Students must visualize how a cut will divide a shape and how the resulting pieces can be further separated. This enhances their ability to mentally manipulate two-dimensional objects \cite{Puig}.
  \item \textbf{Problem-Solving and Optimization:} The core of the problem is finding the minimum number of cuts. This introduces students to the concept of optimization-finding the most efficient solution to a problem. They will engage in trial and error, strategizing different cut placements, and evaluating their effectiveness, thereby developing robust problem-solving skills \cite{Ahdhianto}.
  \item \textbf{Introduction to Graph Theory and Algorithms (Advanced Levels):} For older or more advanced students, the problem can serve as an accessible introduction to graph theory. Each square can be considered a node, and the shared boundaries between squares as edges. The dissection problem then becomes analogous to finding a minimum cut in a graph to isolate all nodes. This can lead to discussions about algorithmic thinking and computational complexity.
  \item \textbf{Perseverance and Resilience:} As with many puzzles, students will encounter challenges and may not find the optimal solution immediately. This fosters perseverance and resilience, encouraging them to persist through difficulties and learn from their mistakes.
\end{enumerate}

\subsubsection{Classroom Activities and Adaptations}
\begin{enumerate}
  \item \textbf{Physical Polyomino Dissection:}
  \begin{enumerate}
      \item \textbf{Materials:} Provide students with polyominoes cut from paper, cardstock, or foam. They can be pre-drawn on a grid to emphasize the unit squares.
      \item \textbf{Activity:} Challenge students to cut the polyominoes into individual squares using the fewest possible cuts. They can record their cuts and the number of pieces obtained after each cut. This hands-on activity allows for direct manipulation and experimentation.
      \begin{enumerate}
          \item \textbf{Discussion:} After the activity, discuss their strategies. Did they find the $N-1$ cuts? Why or why not? What shapes were easier or harder to dissect?
      \end{enumerate}
  \end{enumerate}
\item Digital Interactive Tools:
\begin{enumerate}
    \item \textbf{Software:} Utilize geometry software or online interactive platforms that allow students to draw polyominoes and simulate cuts. This provides a clean, repeatable environment for experimentation and can track the number of cuts automatically.
\item  \textbf{Gamification:} Implement a gamified approach where students earn points for achieving the minimal number of cuts or for dissecting more complex polyominoes. Leaderboards can encourage friendly competition \cite{Cal}.
\end{enumerate}
\item \textbf{Problem-Solving Journals:}
\begin{enumerate}
    \item \textbf{Activity:} Have students document their dissection process in a journal. They can draw the polyomino, sketch their proposed cuts, and explain their reasoning for each cut. This promotes metacognition and allows teachers to assess their thought processes.
\end{enumerate}
\item \textbf{Polyomino Design Challenge:}
\begin{enumerate}
    \item \textbf{Activity:} Challenge students to design a polyomino that is particularly difficult or easy to dissect minimally. They would need to justify their design choices based on the principles of the Minimal Dissection Path Problem.
\end{enumerate}
\item \textbf{Connecting to Real-World Applications:}
\begin{enumerate}
    \item \textbf{Discussion:} Discuss how optimization problems, similar to finding minimal cuts, appear in various real-world scenarios, such as manufacturing (minimizing material waste), logistics (optimizing delivery routes), or computer science (efficient data processing). This helps students see the broader relevance of mathematical thinking.
\end{enumerate}
\end{enumerate}

By integrating the Minimal Dissection Path Problem into the curriculum through these varied activities, educators can transform a seemingly abstract mathematical concept into a tangible, engaging, and highly beneficial learning experience for students of all ages.

\section{Conclusion}
Mathematical folk puzzles, which are strongly steeped in cultural history, are an effective and fascinating teaching tool for students of all ages. This research investigates the educational importance of these puzzles, focusing on their function in developing critical thinking, logical reasoning, perseverance, and collaborative skills. Puzzles promote a good and resilient approach toward problem solving by translating mathematical learning into a playful and engaging experience. The literature delves into a variety of puzzle types, including geometric (Tangrams, dissection puzzles), logic, algebraic, probability (Monty Hall Problem, Birthday Paradox), and combinatorial challenges (Eight Queens Puzzle, Tower of Hanoi). It also explores modern modifications, such as digital and gamified approaches, to improve student involvement and comprehension. Furthermore, a unique idea, the "Minimal Dissection Path Problem for Polyominoes," is introduced and proven, demonstrating that a polyomino of N squares can be dissected into its constituent units using $\mathrm{N}-1$ straight-line cuts. This problem, along with others, provides practical classroom applications that teach fundamental mathematics concepts such as area, spatial reasoning, and optimization, making learning both pleasant and successful.

\end{document}